\pdfoutput=1

\documentclass[letterpaper, 10 pt, conference]{ieeeconf}  %

\IEEEoverridecommandlockouts                              %

\usepackage{graphicx} %
\usepackage{amsmath} %
\usepackage{amssymb}  %
\usepackage{xcolor}   %
\usepackage{algorithm}
\usepackage{algorithmic}

\usepackage{booktabs} %

\DeclareMathOperator*{\argmin}{arg\,min}

\newtheorem{assumption}{Assumption}
\newtheorem{theorem}{Theorem}

\newtheorem{proposition}{Proposition}
\newtheorem{corollary}{Corollary}

\newtheorem{remark}{Remark}
\newtheorem{example}{Example}

\title{\LARGE \bf
Interleaved Projected Gradient Descent for Safe Imitation Learning
}

\author{Shengfan Cao, Francesco Borrelli%
\thanks{Shengfan Cao and Francesco Borrelli are with the University of California, Berkeley,
        {\tt\small shengfan\_cao@berkeley.edu}, {\tt\small fborrelli@berkeley.edu}}%
}

\begin{document}

\maketitle
\thispagestyle{empty}
\pagestyle{empty}

\begin{abstract}

We propose an imitation-learning design for neural-network control policies under state and input constraints. Training alternates a standard imitation gradient step with a block of $k$ safety steps that pull the network's actions toward their projection onto the safe set; at run time, the controller is the trained network alone, with no safety filter. We analyze this scheme as inexact projected gradient descent in the space of policy actions. When the projected actions are recomputed at every safety step and each step moves the actions consistently toward the safe set, letting $k$ grow logarithmically yields asymptotic constraint satisfaction on the training states and bounds the distance to the constrained optimum of the imitation loss; with the projected actions held fixed, the same holds only if they are exactly representable by the network. On a nonlinear autonomous racing task, we compare our method with adding a weighted constraint-violation penalty to the imitation loss. With a sufficiently large weight, our method matches the lap time of unconstrained imitation while reducing the fraction of violating episodes from $15\%$ to $1\%$, about six times fewer than the penalty approach at its best weight. Its lap times are less sensitive to the weight, which instead sets how quickly violations vanish during training. In racing, the safety corrections are sparse and the conditions of the analysis do not hold; the gain arises instead through the data collected during training. These gains come at the cost of additional training computation.

\end{abstract}

\section{Introduction}

Safety filters~\cite{wabersich_predictive_2021,wabersich_datadriven_2023} provide a modular interface between model-based safety analysis and data-driven policies: a minimally invasive, projection-based filter modifies a nominal action only when needed to satisfy the safety constraints, provides formal guarantees under suitable assumptions, and can be paired with an arbitrary policy without retraining the control policy.
This modularity, however, has disadvantages.
An unsafe policy in certain domains triggers frequent or large interventions, and because the filter prioritizes constraint satisfaction over the task objective, the closed loop trajectory can drift far from the intended behavior while the filter is persistently active online.
This raises the question: \textit{can the safety mechanism be used during training to make the nominal policy safer in the first place?}

One line of work, \emph{shielded training}, places a safety mechanism in the learning loop, e.g., a temporal-logic shield~\cite{alshiekh_safe_2018}, a model-predictive backup controller~\cite{bastani_safe_2021}, or Hamilton--Jacobi reachability~\cite{fisac_general_2019}.
Because safety is enforced on the filtered action, the unfiltered policy does not inherit the guarantee, and the shield must be retained at deployment.
A second line \emph{internalizes} constraints into policy optimization.
These methods embed differentiable action projections~\cite{pham_optlayer_2018} or control-barrier-function quadratic programs~\cite{ma_learning_2022} in the policy; our prior work suggests that such exact structure is unnecessary, as a learned safe-set approximation already supplies useful safety gradients~\cite{cao2025cail}.
Approximate model predictive control (MPC) follows a related path, imitating an MPC expert with guarantees from robust design and statistical validation~\cite{hertneck_learning_2018} or from projecting or constraining the network output~\cite{paulson_approximate_2020,chen_approximating_2018}.
Many of these methods combine task and safety losses through a weighted sum, whose weight dictates the safety--performance trade-off and is problem dependent and sensitive to optimization dynamics.

We instead keep task optimization and safety correction as separate operations and interleave them during training.
Similar separations appear in constrained policy optimization: PCPO~\cite{yang_projection_2020} alternates objective and constraint updates for expected-cost constraints, SCPO~\cite{cao_constrained_2025} projects onto a safe set estimated in weight space from rollouts, and FOCOPS~\cite{zhang_focops_2020} solves for a constrained update in policy space before regressing the parameters toward it, a structure our safety block shares.
In contrast, we consider pointwise action constraints on a fixed imitation-learning dataset, with safety correction implemented as regression toward projected actions.

Our contributions are:
\begin{itemize}
    \item An interleaved training scheme that alternates performance steps with blocks of safety steps, avoiding scalarization of task and safety objectives.
    \item An inexact projected gradient descent (PGD) analysis in the space of policy outputs: safety blocks approximate projections with geometrically decaying error (Proposition~\ref{prop:geometric-projection-approx}), and a logarithmic safety schedule yields asymptotic feasibility on the training states and a distance bound to the constrained optimum, with fixed (Proposition~\ref{prop:inexact-pgd}) or refreshed (Theorem~\ref{thm:refreshed}) projection targets.
    \item An empirical study on a constrained optimal control example and on nonlinear safe imitation learning for autonomous racing, comparing the safety--performance trade-off and hyperparameter sensitivity against weighted-gradient, primal--dual, and safety-correction baselines.
\end{itemize}

\section{Problem Formulation}
\label{sec:problem-formulation}
Consider a discrete-time, time-invariant system 
\begin{equation}
    x_{\tau+1} = f(x_\tau, u_\tau),\quad \forall \tau \in \mathbb{N},
\end{equation}
subject to persistent state and input constraints
\begin{equation}
    x_\tau \in \mathcal{X},\quad u_\tau \in \mathcal{U},\quad \forall \tau \in \mathbb{N}.
\end{equation}
Let $\mathcal{K} \subseteq \mathcal{X}$ denote the maximal control invariant set, i.e., the largest set satisfying
\begin{equation}
    \mathcal{K} = \{ x \in \mathcal{X} \mid \exists u \in \mathcal{U},\ f(x, u) \in \mathcal{K}\}.
\end{equation}
Equivalently, $\mathcal{K}$ is the set of states from which some admissible input sequence keeps $x_\tau \in \mathcal{X}$ for all $\tau$.

We want to train a nonlinear parametric feedback policy $\pi_\theta$ through behavior cloning (BC). 
For the analysis, demonstrations are given by a fixed dataset $\mathcal{D} = \{(x_i, \bar{u}_i)\}_{i=1}^{N}$, and we denote its set of training states by $\mathcal{D}_X = \{x_i\}_{i=1}^{N}$. 
Assume $\mathcal{D}_X \subseteq \mathcal{K}$.
However, the demonstrations themselves need not be safe. In particular, it is not assumed that $f(x_i, \bar{u}_i) \in \mathcal{K},~\forall (x_i, \bar{u}_i) \in \mathcal{D}$. 

We consider the behavior cloning loss
\begin{equation}\label{eq:bc-loss}
    J(\theta) = \frac{1}{2N} \sum_{i} \bigl[ \Vert \pi_\theta(x_i) - \bar{u}_i \Vert_2^2 \bigr].
\end{equation}
For a nonlinear parameterization $\pi_\theta$, $J$ is generally nonconvex
as a function of $\theta$.

In addition, assume that a projection-based safety filter is given.
For every $x \in \mathcal{K}$, let $\mathcal{C}(x) \subseteq \mathcal{U}$ be a state-dependent, nonempty, closed, and convex set, satisfying
\begin{equation}
    u \in \mathcal{C}(x) \quad \Rightarrow \quad f(x, u) \in \mathcal{K}. 
\end{equation}
The safety filter performs a Euclidean projection of a proposed action
\begin{equation}\label{eq:safety-filter}
    \Pi_{\mathcal{C}(x)}(u) = \argmin_{u^{\text{proj}} \in \mathcal{C}(x)} \Vert u^{\text{proj}} - u \Vert_2^2,
\end{equation}

The corresponding feasible parameter set on the training states is \footnote{Persistent closed-loop safety additionally requires an initial state in $\mathcal{K}$ and satisfaction of the safe-action condition at every visited state under valid dynamics and safety constraints.}
\begin{equation}\label{eq:parameter-feasible-set}
    \Theta_{\text{safe}}
    \triangleq
    \left\{
    \theta:
    \pi_\theta(x_i)\in\mathcal C(x_i),
    \ \forall x_i \in \mathcal{D}_X
    \right\}.
\end{equation}
Our goal is to find a local solution of 
\begin{equation}\label{eq:constrained-policy-opt}
    \min_{\theta \in \Theta_{\text{safe}}} J(\theta).
\end{equation}
Define the squared distance $d_{\mathcal C}(v)=\|v-\Pi_{\mathcal C}(v)\|_2^2$ and the safety loss 
\begin{equation}\label{eq:safety-loss}
    \Phi_{\mathcal{C}}(\theta) = \frac{1}{2N} \sum_{i=1}^{N} d_{\mathcal{C}(x_i)}(\pi_\theta(x_i)).
\end{equation}

For differentiable $\pi_\theta$, its gradient is
\begin{equation}\label{eq:safety-loss-gradient}
    \nabla_\theta\Phi_{\mathcal C}(\theta)
    =
    \frac{1}{N}
    \sum_{i=1}^N
    D_\theta\pi_\theta(x_i)^\top
    \left(
        \pi_\theta(x_i)
        -
        \Pi_{\mathcal C(x_i)}
        \bigl(\pi_\theta(x_i)\bigr)
    \right),
\end{equation}
where $D_\theta \pi_\theta(x_i)$ denotes the Jacobian of the policy output $\pi_\theta(x_i)$ with respect to the parameters $\theta$. 

Define the performance and safety gradient operators
\begin{align}
    P_\eta(\theta)
    &\triangleq\theta-\eta\nabla J(\theta),\\
    S_\beta(\theta)
    &\triangleq\theta-\beta\nabla\Phi_{\mathcal C}(\theta).
\end{align}
Both operators act in parameter space.

If we were to solve~\eqref{eq:constrained-policy-opt} with projected gradient descent in the weight space, we would need a projection operator onto the feasible parameter set,
\begin{equation}\label{eq:projection-onto-weight-space-safe-set}
    \theta_{t+1}
    \in
    \Pi_{\Theta_{\rm safe}}
    \left(
        P_\eta(\theta_t)
    \right).
\end{equation}
However, 
although each action-space set $\mathcal C(x_i)$ is convex, the inverse image $\Theta_{\rm safe}$ is generally nonconvex for a nonlinear policy.
Consequently, projection onto $\Theta_{\rm safe}$ may be nonunique and computationally difficult.

Moreover, a gradient step on~\eqref{eq:safety-loss} does not in general
induce independent projections of the policy outputs toward their
action-space safe sets. To see this, for each $x_i \in \mathcal{D}_X$, let $r_i(\theta)=\pi_\theta(x_i)-\Pi_{\mathcal C(x_i)}(\pi_\theta(x_i))$ and $G_i(\theta) = D_\theta\pi_\theta(x_i)$.
Assuming the policy Jacobian is locally Lipschitz on the update neighborhood, a first-order expansion gives
\begin{equation}
\begin{aligned}
    \pi_{S_\beta(\theta)}(x_i)
    =
    \pi_\theta(x_i)
    -
    \frac{\beta}{N}
    \sum_{j=1}^N
    G_i(\theta)G_j(\theta)^\top r_j(\theta)
    +O(\beta^2).
\end{aligned}
\end{equation}
Thus, shared policy parameters couple the safety corrections across
training states, and the induced action-space displacement generally
differs from the direct soft-projection direction $-r_i(\theta)$.

Motivated by this observation, we instead interleave a single \emph{performance step} $P_\eta$ with a \emph{safety block} of $K_t$ \emph{safety steps}, where $K_t$ denotes the number of safety steps performed at iteration $t$ and increases according to a prescribed schedule. 
We analyze the block in Sec.~\ref{sec:inexact-pgd} as an approximate projection within inexact projected gradient descent.
The projection-approximation and primal-consistency conditions below connect this correction to constrained optimization.

\section{Inexact Projected Gradient Descent}
\label{sec:inexact-pgd}

We now interpret interleaved training as an \textit{inexact projected gradient method} in the policy-output space. The analysis is idealized: it assumes a fixed dataset, convex safe action sets with exact Euclidean projections, and relative-error conditions on how parameter updates move the outputs. It clarifies the roles of the performance step, the safety block, its schedule, and target refreshing; Sec.~\ref{sec:mechanism-ablation} discusses how closely a trained racing policy meets these conditions.
The analysis bounds the safety block's projection error and propagates it, together with relative distortion in the performance step, into an asymptotic bound on the distance to the constrained optimum.

Let
\(\mathcal{D}_X=\{x_i\}_{i=1}^N\) denote the fixed state dataset and define the stacked policy output
\begin{equation}
    v(\theta)
    = 
    \operatorname{col}\!\left(
    \pi_\theta(x_1),\ldots,\pi_\theta(x_N)
    \right).
\end{equation}
Write $v_t=v(\theta_t)$ for the output at outer iteration $t$.
For each state \(x_i\), let \(\mathcal{C}(x_i)\) be a nonempty, closed, and convex safe action set, and define
\begin{equation}
    \mathcal{C}_{\mathcal D}
    = 
    \mathcal{C}(x_1)\times\cdots\times\mathcal{C}(x_N).
\end{equation}
Equip output space with the normalized Euclidean norm
\begin{equation}
    \|v\|_{\mathcal D}^2
    = 
    \frac{1}{N}\sum_{i=1}^N \|v_i\|_2^2.
\end{equation}
Output-space gradients use the inner product $\langle a,b\rangle_{\mathcal D}=a^\top b/N$; parameter-space gradients are Euclidean.
Let $\bar{u} = \text{col}(\bar{u}_1, \dots, \bar{u}_N)$ and define
\begin{equation}
    \ell(v) = \frac{1}{2}\| v - \bar{u} \|_\mathcal{D}^2, \qquad J(\theta) = \ell(v(\theta)).
    \label{eq:output-bc-loss}
\end{equation}
Projection onto \(\mathcal{C}_{\mathcal D}\) decomposes pointwise:
\begin{equation}
    [\Pi_{\mathcal{C}_{\mathcal D}}(v)]_i
    =
    \Pi_{\mathcal C(x_i)}(v_i).
    \label{eq:dataset-projection}
\end{equation}

The nonlinear parameterization \(\theta\mapsto v(\theta)\) generally prevents a single weight update from realizing this projection exactly. We therefore approximate the projected policy by supervised regression.

Given \(\theta_{\mathrm{ref}}\), define the fixed projected targets
\begin{equation}
    \hat{u}_i
    =
    \Pi_{\mathcal C(x_i)}
    \bigl(\pi_{\theta_{\mathrm{ref}}}(x_i)\bigr)
\end{equation}
and the corresponding regression objective
\begin{equation}
    R_{\theta_{\mathrm{ref}}}(\theta)
    =
    \frac{1}{2N}
    \sum_{i=1}^N
    \|\pi_\theta(x_i)-\hat{u}_i\|_2^2.
    \label{eq:fixed-projection-loss}
\end{equation}
Unlike the safety loss \(\Phi_{\mathcal C}\), whose projection targets change with \(\theta\), the targets in \eqref{eq:fixed-projection-loss} are frozen throughout one safety block. Starting from \(\theta^{(0)}=\theta_{\mathrm{ref}}\), the block performs \(K\) gradient steps
\begin{equation}
    \theta^{(k+1)}
    =
    \theta^{(k)}
    -
    \beta
    \nabla R_{\theta_{\mathrm{ref}}}
    \bigl(\theta^{(k)}\bigr),
    \qquad k=0,\ldots,K-1,
    \label{eq:safety-regression-step}
\end{equation}
and we write
\begin{equation}
    \hat S_{K}(\theta_{\mathrm{ref}})
    \triangleq
    \theta^{(K)}
    \label{eq:fixed-target-block}
\end{equation}
for the resulting \emph{fixed-target safety block}; the hat matches the frozen targets \(\hat u_i\). One iteration of the scheme is then \(\theta_{t+1}=\hat S_{K_t}(P_\eta(\theta_t))\). \footnote{Note that \(K_t\) applications of the safety gradient operator of Sec.~\ref{sec:problem-formulation}, \(S_\beta^{K_t}(P_\eta(\theta_t))\) recomputes the projected actions at every step; Sec.~\ref{sec:refreshed} analyzes this variant.}

\subsection{Fixed Projection Targets}

Let
\begin{equation}
    G(\theta)
    =
    \operatorname{col}\bigl(G_1(\theta),\ldots,G_N(\theta)\bigr),
    \quad
    G_i(\theta)=D_\theta\pi_\theta(x_i),
\end{equation}
be the stacked policy Jacobian, and let
\begin{equation}
    H(\theta)=\frac1N G(\theta)G(\theta)^\top
\end{equation}
be the empirical neural tangent kernel (NTK) Gram matrix
\cite{jacot_neural_2018}.
For a fixed target, write
\begin{equation}
    r(\theta)
    =
    v(\theta)
    -
    \Pi_{\mathcal C_{\mathcal D}}\bigl(v(\theta_{\mathrm{ref}})\bigr),
    ~~
    R(\theta)
    =
    R_{\theta_{\mathrm{ref}}}(\theta)
    =
    \frac12\|r(\theta)\|_{\mathcal D}^2.
\end{equation}
Since $\nabla R(\theta)=N^{-1}G(\theta)^\top r(\theta)$,
for $r(\theta)\ne0$ we have
\begin{equation}
    \|\nabla R(\theta)\|_2^2
    =
    2\mu(\theta)R(\theta),
    \qquad
    \mu(\theta)
    =
    \frac{r(\theta)^\top H(\theta)r(\theta)}
         {\|r(\theta)\|_2^2}.
    \label{eq:restricted-pl}
\end{equation}
Thus $\mu(\theta)\ge\underline\mu>0$ is a Polyak--{\L}ojasiewicz (PL$^*$) condition
relative to zero \cite{liu_loss_2022}, imposed below only
along the block iterates.

\begin{proposition}[Geometric projection approximation]
\label{prop:geometric-projection-approx}
Fix an outer iteration $t$ and the fixed-target safety block
\eqref{eq:safety-regression-step}--\eqref{eq:fixed-target-block}.
Suppose $\nabla R$ is $L_S$-Lipschitz on a region containing
the update segments, and
\[
    0<\beta\le L_S^{-1},
    \qquad
    0<\beta\underline\mu<1.
\]
For $k=0,\ldots,K-1$, write $r_k=r(\theta^{(k)})$ and assume
$\mu_k=\mu(\theta^{(k)})\ge\underline\mu$ whenever $r_k\ne0$;
set $\mu_k=\underline\mu$ when $r_k=0$.
Then
\begin{equation}
    \|r_K\|_{\mathcal D}
    \le
    \left(\prod_{k=0}^{K-1}\sqrt{1-\beta\mu_k}\right)
    \|r_0\|_{\mathcal D}
    \le
    q^K\|r_0\|_{\mathcal D},
    \label{eq:geometric-projection-error}
\end{equation}
where $q=\sqrt{1-\beta\underline\mu}\in(0,1)$.
\end{proposition}

\begin{proof}
Write $R_k=R(\theta^{(k)})$.
The descent lemma and \eqref{eq:restricted-pl} give
\[
    R_{k+1}
    \le
    R_k-\frac{\beta}{2}\|\nabla R(\theta^{(k)})\|_2^2
    =
    (1-\beta\mu_k)R_k.
\]
If $r_k=0$, then $\nabla R(\theta^{(k)})=0$, so subsequent
iterates remain unchanged and the same inequality holds.
When $r_k\ne0$, nonnegativity of $R_{k+1}$ ensures
$1-\beta\mu_k\ge0$.
Iterating, using $2R_k=\|r_k\|_{\mathcal D}^2$, and taking
square roots proves \eqref{eq:geometric-projection-error}.
\end{proof}

A uniform $\underline\mu$ with $K\to\infty$ forces $R\to0$, so the hypotheses of Proposition~\ref{prop:geometric-projection-approx} implicitly require the frozen targets $\Pi_{\mathcal C_{\mathcal D}}(v(\theta_{\mathrm{ref}}))$ to be exactly representable by the network. When they are not, the residual approaches the orthogonal complement of the range of $G$, $\mu_k\to0$, and a fixed-target block can stall outside the feasible set even when $\Theta_{\rm safe}\neq\emptyset$.

The logarithmic schedule follows directly from the geometric convergence of the safety block. Let \(z_t=v(P_\eta(\theta_t))\) denote the policy output after the \(t\)-th performance step, and define the safety-block projection error
\begin{equation}
    \epsilon_t=\bigl\|v_{t+1}-\Pi_{\mathcal C_{\mathcal D}}(z_t)\bigr\|_{\mathcal D}.
    \label{eq:safety-block-error}
\end{equation}
 
\begin{corollary}[Logarithmic safety schedule]
\label{cor:log-schedule}
Suppose the hypotheses of Proposition~\ref{prop:geometric-projection-approx} hold at every outer iteration with constants \(L_S\), \(\underline\mu\), and \(\beta\) independent of \(t\), and that \(\|z_t-\Pi_{\mathcal C_{\mathcal D}}(z_t)\|_{\mathcal D}\le B\) for all \(t\). If \(K_t=\lceil c\log(t+2)\rceil\) with \(c>0\), then
\begin{equation}
    \epsilon_t\le B\,(t+2)^{-\gamma},
    \qquad \gamma=c\log(1/q)>0,
    \label{eq:log-schedule-rate}
\end{equation}
so \(\epsilon_t\to0\).
\end{corollary}
 
\begin{proof}
At iteration \(t\), the block starts from \(\theta_{\mathrm{ref}}=P_\eta(\theta_t)\), so \(r_0=z_t-\Pi_{\mathcal C_{\mathcal D}}(z_t)\) and \(r_{K_t}=v_{t+1}-\Pi_{\mathcal C_{\mathcal D}}(z_t)\). Proposition~\ref{prop:geometric-projection-approx} with a common \(q\in(0,1)\) gives \(\epsilon_t\le q^{K_t}B\le q^{c\log(t+2)}B=B(t+2)^{-\gamma}\).
\end{proof}

We next connect the preceding result to projected gradient descent. 
Consider the constrained policy-output problem
\begin{equation}
    \min_{v\in\mathcal C_{\mathcal D}} \ell(v).
    \label{eq:output-constrained-problem}
\end{equation}
The behavior cloning objective~\eqref{eq:output-bc-loss} satisfies
$\nabla\ell(v)=v-\bar u$ and $L_P=\mu_P=1$.
We model the discrepancy between parameter and output gradient steps
through a relative-error condition, as in constrained optimization with
inexact gradient oracles~\cite{hallak_relative_2024}.

\begin{assumption}[Relative primal-step consistency]
\label{asm:primal-step-consistency}
Let $v_t=v(\theta_t)$ and let $z_t$ be the output after the performance
step. There exist $\kappa\in[0,1)$, effective output step sizes
$\alpha_t\in[\alpha_{\min},\alpha_{\max}]$ with
$0<\alpha_{\min}\le\alpha_{\max}\le1$, and $\sigma_t\ge0$ with
$\sigma_t\to0$, such that
\begin{align}
    z_t&=v_t-\alpha_t\nabla\ell(v_t)+\delta_t,
    \label{eq:inexact-primal-step}\\
    \|\delta_t\|_{\mathcal D}
    &\le\kappa\alpha_t\|\nabla\ell(v_t)\|_{\mathcal D}+\sigma_t.
    \label{eq:relative-primal-error}
\end{align}
\end{assumption}
The effective step size $\alpha_t$ accounts for the output geometry induced by $P_\eta$. For affine outputs $v(\theta)=A\theta+b$ with $mI\preceq H\preceq MI$, $m>0$, and $s\eta\le1$ for $s=(M+m)/2$, the assumption holds with $\alpha_t=s\eta$, $\kappa=(M-m)/(M+m)$, and $\sigma_t=0$, since $\delta_t=-\eta(H-sI)\nabla\ell(v_t)$.

Let \(v^\star=\Pi_{\mathcal C_{\mathcal D}}(\bar u)\) be the unique solution of~\eqref{eq:output-constrained-problem}, \(d_\star=\operatorname{dist}_{\mathcal D}(\bar u,\mathcal C_{\mathcal D})\) the distance of the demonstrations from the feasible set, and \(e_t=\|v_t-v^\star\|_{\mathcal D}\).

\begin{proposition}[Asymptotic error of interleaved training]
\label{prop:inexact-pgd}
Suppose Assumption~\ref{asm:primal-step-consistency} holds and \(\limsup_{t\to\infty}\epsilon_t\le\bar\epsilon\), with \(\epsilon_t\) as in~\eqref{eq:safety-block-error}. Then \(\limsup_{t\to\infty}\operatorname{dist}_{\mathcal D}(v_t,\mathcal C_{\mathcal D})\le\bar\epsilon\) and
\begin{equation}
    \limsup_{t\to\infty}e_t
    \le\frac{\kappa\,d_\star}{1-\kappa}
    +\frac{\bar\epsilon}{\alpha_{\min}(1-\kappa)}.
    \label{eq:asymptotic-output-error}
\end{equation}
In particular, \(v_t\to v^\star\) if \(\bar\epsilon=0\) and \(\kappa d_\star=0\).
\end{proposition}
 
\begin{proof}
This is the standard analysis of projected gradient descent with relative gradient error and inexact projection~\cite{hallak_relative_2024,scm2011}: since $v^\star=\Pi_{\mathcal C_{\mathcal D}}((1-\alpha_t)v^\star+\alpha_t\bar u)$, nonexpansiveness and $\|\nabla\ell(v_t)\|_{\mathcal D}\le e_t+d_\star$ give $e_{t+1}\le[1-\alpha_t(1-\kappa)]e_t+\kappa\alpha_t d_\star+\sigma_t+\epsilon_t$, which iterates to~\eqref{eq:asymptotic-output-error}, and $\operatorname{dist}_{\mathcal D}(v_{t+1},\mathcal C_{\mathcal D})\le\epsilon_t$.
\end{proof}

Under Corollary~\ref{cor:log-schedule}, $\bar\epsilon=0$, the outputs approach the feasible set on the training states, at asymptotic distance at most $\kappa d_\star/(1-\kappa)$ from the optimum, while a constant block length $K$ gives $\bar\epsilon\le Bq^K$.

\subsection{Refreshed Projection Targets}
\label{sec:refreshed}
Proposition~\ref{prop:inexact-pgd} needs the projected safe action from every intermediate output $z_t$ to be exactly representable, which a finite network fitting many states rarely achieves. Recomputing the targets at every safety step removes this requirement.

The refreshed variant replaces $\hat S_{K_t}$ with $S_\beta^{K_t}$, i.e., gradient descent on $\Phi_{\mathcal C}(\theta)=\frac12\operatorname{dist}_{\mathcal D}^2(v(\theta),\mathcal C_{\mathcal D})$, recomputing the projected actions at every safety step. Starting from $\theta^{(0)}=P_\eta(\theta_t)$, write $v^{(k)}=v(\theta^{(k)})$, $p_k=\Pi_{\mathcal C_{\mathcal D}}(v^{(k)})$, and $\tilde r_k=v^{(k)}-p_k$, so that $v^{(0)}=z_t$ and $v^{(K_t)}=v_{t+1}$. We model each safety step through a relative-error condition analogous to Assumption~\ref{asm:primal-step-consistency}.

\begin{assumption}[Relative safety-step consistency]
\label{asm:safety-step-consistency}
There exist $\kappa_S\in[0,1)$ and $b_{\min}\in(0,1]$ such that every safety step with $\tilde r_k\neq0$ satisfies
\begin{equation}
    v^{(k+1)}=v^{(k)}-b_k\tilde r_k+\xi_k,
    \qquad
    \|\xi_k\|_{\mathcal D}\le\kappa_S b_k\|\tilde r_k\|_{\mathcal D},
    \label{eq:safety-step-consistency}
\end{equation}
for some $b_k\in[b_{\min},1]$.
\end{assumption}

Minimizing over $b_k$, the smallest admissible $\kappa_S$ for a step is the sine of the angle between $v^{(k)}-v^{(k+1)}$ and $\tilde r_k$.
To first order in $\beta$, this is the angle between $H(\theta^{(k)})\tilde r_k$ and $\tilde r_k$. Since $\tilde r_k$ lies in the normal cone of $\mathcal C_{\mathcal D}$ at $p_k$, the condition concerns only directions that point out of the feasible set.

\begin{proposition}[Refreshed safety block]
\label{prop:refreshed-block}
Under Assumption~\ref{asm:safety-step-consistency}, for every $K$,
\begin{align}
    \operatorname{dist}_{\mathcal D}(v^{(K)},\mathcal C_{\mathcal D})
    &\le\rho^K\operatorname{dist}_{\mathcal D}(v^{(0)},\mathcal C_{\mathcal D}),
    \label{eq:refreshed-contraction}\\
    \|p_K-p_0\|_{\mathcal D}
    &\le\frac{\kappa_S}{1-\kappa_S}\operatorname{dist}_{\mathcal D}(v^{(0)},\mathcal C_{\mathcal D}),
    \label{eq:refreshed-drift}
\end{align}
where $\rho=1-b_{\min}(1-\kappa_S)<1$.
\end{proposition}
\begin{proof}
For $b\in[0,1]$, $v^{(k)}-b\tilde r_k$ lies on the segment from $v^{(k)}$ to $p_k$ and hence projects to $p_k$. Nonexpansiveness of the projection then gives $\|p_{k+1}-p_k\|_{\mathcal D}\le\|\xi_k\|_{\mathcal D}\le\kappa_S b_k\|\tilde r_k\|_{\mathcal D}$, and
$\|\tilde r_{k+1}\|_{\mathcal D}\le\|v^{(k+1)}-p_k\|_{\mathcal D}=\|(1-b_k)\tilde r_k+\xi_k\|_{\mathcal D}\le\bigl(1-b_k(1-\kappa_S)\bigr)\|\tilde r_k\|_{\mathcal D}$, which proves~\eqref{eq:refreshed-contraction}. With $\pi_k=\prod_{j<k}(1-b_j(1-\kappa_S))$, summing the drift bound gives $\kappa_S\|\tilde r_0\|_{\mathcal D}\sum_k b_k\pi_k=\frac{\kappa_S}{1-\kappa_S}\|\tilde r_0\|_{\mathcal D}(1-\pi_K)$, which proves~\eqref{eq:refreshed-drift}. A zero residual stops the block.
\end{proof}

Proposition~\ref{prop:refreshed-block} is what refreshing buys. With fixed targets, more safety steps cannot help once the remaining correction lies outside what the network can represent, so the output may stay infeasible however many steps the block takes. A refreshed block needs only a realizable component along each outward normal, so its residual keeps contracting, at the price of a drift~\eqref{eq:refreshed-drift} proportional to the constraint violation.

\begin{theorem}[Interleaving with refreshed targets]
\label{thm:refreshed}
Let $\theta_{t+1}=S_\beta^{K_t}(P_\eta(\theta_t))$, where $K_t$ are positive integers with $K_t\to\infty$. Suppose Assumption~\ref{asm:primal-step-consistency} holds, Assumption~\ref{asm:safety-step-consistency} holds throughout every safety block with common constants $\kappa_S\in[0,1)$ and $b_{\min}\in(0,1]$, and $\operatorname{dist}_{\mathcal D}(z_t,\mathcal C_{\mathcal D})\le B<\infty$ for all $t$. Then, with $\rho=1-b_{\min}(1-\kappa_S)$, $\operatorname{dist}_{\mathcal D}(v_{t+1},\mathcal C_{\mathcal D})\le\rho^{K_t}B\to0$. If, moreover, $\kappa+\kappa_S<1$, then
\begin{equation}
\limsup_{t\to\infty}e_t\le\frac{\kappa+\kappa_S}{1-\kappa-\kappa_S}\,d_\star;
\label{eq:refreshed-asymptotic-error}
\end{equation}
in particular, under this condition $v_t\to v^\star$ if $d_\star=0$.
\end{theorem}

\begin{proof}
Write $D_t=\operatorname{dist}_{\mathcal D}(z_t,\mathcal C_{\mathcal D})$, $q_t=\rho^{K_t}$, and $c_S=\kappa_S/(1-\kappa_S)$. Proposition~\ref{prop:refreshed-block} gives $\operatorname{dist}_{\mathcal D}(v_{t+1},\mathcal C_{\mathcal D})\le q_tD_t\le q_tB\to0$, which proves feasibility. For the optimality bound, the drift bound~\eqref{eq:refreshed-drift} and the triangle inequality bound the block error $\epsilon_t=\|v_{t+1}-\Pi_{\mathcal C_{\mathcal D}}(z_t)\|_{\mathcal D}$ of Proposition~\ref{prop:inexact-pgd} as
\begin{equation}
\epsilon_t\le\operatorname{dist}_{\mathcal D}(v_{t+1},\mathcal C_{\mathcal D})+c_SD_t\le q_tB+c_SD_t.
\label{eq:refreshed-block-error}
\end{equation}
Since $\nabla\ell(v_t)=v_t-\bar u$, Assumption~\ref{asm:primal-step-consistency} gives $z_t=(1-\alpha_t)v_t+\alpha_t\bar u+\delta_t$ with $\|\delta_t\|_{\mathcal D}\le\kappa\alpha_t(e_t+d_\star)+\sigma_t$. Distance to a nonempty closed convex set is convex and $1$-Lipschitz; hence, using $0<\alpha_t\le1$ and, for $t\ge1$, the feasibility estimate,
\begin{align}
D_t
&\le(1-\alpha_t)\operatorname{dist}_{\mathcal D}(v_t,\mathcal C_{\mathcal D})+\alpha_td_\star+\|\delta_t\|_{\mathcal D}
\nonumber\\
&\le\kappa\alpha_te_t+\alpha_t(1+\kappa)d_\star+\sigma_t+(1-\alpha_t)q_{t-1}B.
\label{eq:refreshed-preblock-distance}
\end{align}
Substituting~\eqref{eq:refreshed-block-error} and~\eqref{eq:refreshed-preblock-distance} into the one-step estimate $e_{t+1}\le[1-\alpha_t(1-\kappa)]e_t+\kappa\alpha_td_\star+\sigma_t+\epsilon_t$ of Proposition~\ref{prop:inexact-pgd}, and using $1+c_S=(1-\kappa_S)^{-1}$, yields
\[
e_{t+1}\le(1-a\alpha_t)e_t+\alpha_t\frac{\kappa+\kappa_S}{1-\kappa_S}d_\star+\tau_t,
\qquad a=\frac{1-\kappa-\kappa_S}{1-\kappa_S},
\]
where $\tau_t=(1+c_S)\sigma_t+Bq_t+c_S(1-\alpha_t)Bq_{t-1}\to0$. With $E=(\kappa+\kappa_S)d_\star/(1-\kappa-\kappa_S)$, this reads $e_{t+1}-E\le(1-a\alpha_t)(e_t-E)+\tau_t$. Since $0<a\le1$ and $0<\alpha_{\min}\le\alpha_t\le1$, $w_t=\max\{e_t-E,0\}$ satisfies $w_{t+1}\le(1-a\alpha_{\min})w_t+\tau_t$, so $w_t\to0$, which proves~\eqref{eq:refreshed-asymptotic-error}.
\end{proof}

Refreshing thus converts the absolute projection error $\bar\epsilon$ of Proposition~\ref{prop:inexact-pgd}, which vanishes only if the frozen targets are realizable, into a relative error: it permits the projected targets to drift while preserving asymptotic feasibility. Under the stated uniform consistency and boundedness assumptions, feasibility holds for every $\kappa_S<1$, whereas the optimality bound requires $\kappa+\kappa_S<1$. In particular, when the primal step is exact ($\kappa=0$), the same condition $\kappa_S<1$ suffices for both conclusions, with $\limsup_te_t\le\kappa_Sd_\star/(1-\kappa_S)$. The bound remains conservative because it accumulates projection drift through norm inequalities.

\begin{remark}[Relation to BC on projected demonstrations]
\label{rem:bc-projected}
Because $\ell$ and $\mathcal C_{\mathcal D}$ decompose across samples, $v^\star=\Pi_{\mathcal C_{\mathcal D}}(\bar u)$, which plain BC on projected demonstrations fits in the least-squares sense. When $v^\star$ is not exactly representable, that fitting error is indifferent to the constraints and can leave $\mathcal C_{\mathcal D}$, whereas refreshed safety steps require only a realizable component along each outward normal (Proposition~\ref{prop:refreshed-block}). The safety regression also acts on bank states that carry no demonstration.
\end{remark}

Algorithm~\ref{alg:log-interleaving} summarizes both variants. With the projected targets recomputed once after each performance step and held fixed during the safety block, it computes $\theta_{t+1}=\hat S_{K_t}(P_\eta(\theta_t))$ (Proposition~\ref{prop:inexact-pgd}); recomputing them before every safety step (line~7) gives the refreshed variant $\theta_{t+1}=S_\beta^{K_t}(P_\eta(\theta_t))$ (Theorem~\ref{thm:refreshed}).

\begin{algorithm}[t]
\caption{Log-Interleaved Projected-Target Training}
\label{alg:log-interleaving}
\begin{algorithmic}[1]
\REQUIRE Initial weights \(\theta_0\); learning rates \(\eta,\beta\); schedule constant \(c\)
\FOR{\(t=0,1,2,\ldots\)}
    \STATE \textbf{Performance step:} \(\theta_t^P \leftarrow P_\eta(\theta_t) = \theta_t - \eta\nabla_\theta J(\theta_t)\)
    \STATE \textbf{Freeze reference:} \(\theta_{\mathrm{ref}} \leftarrow \theta_t^P\), \(\ \theta \leftarrow \theta_{\mathrm{ref}}\)
    \STATE \textbf{Compute projected targets:} \(\hat{u}_i \leftarrow \Pi_{\mathcal C(x_i)}\bigl(\pi_{\theta_{\mathrm{ref}}}(x_i)\bigr)\) for \(i=1,\ldots,N\)
    \STATE \(K_t \leftarrow \lceil c\log(t+2)\rceil\)
    \FOR[safety block]{\(k=0,\ldots,K_t-1\)}
        \STATE \textit{(Refreshed variant:} \(\hat{u}_i \leftarrow \Pi_{\mathcal C(x_i)}\bigl(\pi_{\theta}(x_i)\bigr)\)\textit{)}
        \STATE \textbf{Safety step:} \(\theta \leftarrow \theta - \beta\nabla_\theta R_{\theta_{\mathrm{ref}}}(\theta)\), \ where \(R_{\theta_{\mathrm{ref}}}(\theta) = \dfrac{1}{2N}\sum_{i=1}^{N}\|\pi_\theta(x_i)-\hat{u}_i\|_2^2\)
    \ENDFOR
    \STATE \(\theta_{t+1}\leftarrow\theta\) \COMMENT{\(\hat S_{K_t}(\theta_t^P)\) or \(S_\beta^{K_t}(\theta_t^P)\)}
\ENDFOR
\end{algorithmic}
\end{algorithm}

\section{Connection with Fused Safety Penalties}
\label{sec:fused-connection}

The output-space formulation also clarifies the effect of combining performance and safety gradients. Define
\begin{equation}
\varphi(v)=\frac12\|v-\Pi_{\mathcal C_{\mathcal D}}(v)\|_{\mathcal D}^2.
\end{equation}
Since $\nabla\varphi(v)=v-\Pi_{\mathcal C_{\mathcal D}}(v)$, a fused output update is
\begin{equation}
U_{\eta,\rho}(v)=v-\eta\nabla\ell(v)-\rho\nabla\varphi(v).
\label{eq:fused-update}
\end{equation}
For $\eta>0$ and $\rho\ge0$, setting $\lambda=\rho/\eta$ gives
\begin{equation}
U_{\eta,\rho}(v)=v-\eta\nabla\ell_\lambda(v),\qquad
\ell_\lambda(v)=\ell(v)+\lambda\varphi(v).
\label{eq:fused-penalty-equivalence}
\end{equation}
Thus, the relative gradient weight defines a quadratic-penalty relaxation of the constrained problem~\cite{nocedal_numerical_2006}.

\begin{proposition}[Finite-penalty bias]
\label{prop:fused-bias}
Let $\mathcal C_{\mathcal D}$ be nonempty, closed, and convex, and let $\ell$ be differentiable. If a constrained optimizer $v^\star\in\argmin_{v\in\mathcal C_{\mathcal D}}\ell(v)$ satisfies $\nabla\ell(v^\star)\ne0$, then it is not a stationary point of $\ell_\lambda$ for any finite $\lambda\ge0$.
\end{proposition}
\begin{proof}
Feasibility gives $\Pi_{\mathcal C_{\mathcal D}}(v^\star)=v^\star$, so $\nabla\varphi(v^\star)=0$ and $\nabla\ell_\lambda(v^\star)=\nabla\ell(v^\star)\ne0$.
\end{proof}

At an active constrained optimum, the performance gradient can be balanced by the normal cone of the feasible set. The squared-distance penalty instead has zero gradient throughout that set. For $\lambda>0$, stationarity of $\ell_\lambda$ implies
\begin{equation}
\operatorname{dist}_{\mathcal D}(v_\lambda,\mathcal C_{\mathcal D})
=\frac{\|\nabla\ell(v_\lambda)\|_{\mathcal D}}{\lambda}.
\label{eq:fused-bias-rate}
\end{equation}
A uniform bound on the performance gradient therefore gives an $O(1/\lambda)$ residual bound along these stationary points.

Increasing $\lambda$ also changes the step-size requirement. For convex $L_P$-smooth $\ell$, the objective $\ell_\lambda$ is at most $(L_P+\lambda)$-smooth, so a sufficient condition for the fused gradient step to be nonexpansive is
\begin{equation}
0<\eta<\frac{2}{L_P+\lambda},\qquad
\eta L_P+\rho<2.
\label{eq:fused-stability}
\end{equation}
Penalty continuation can increase $\lambda$ while decreasing $\eta$, at the cost of another schedule and potentially poorer conditioning.

Fixed-target interleaving gives the correction budget a separate role: $K_t$ controls projection accuracy while the performance update is held fixed. Under Corollary~\ref{cor:log-schedule}, logarithmic growth of $K_t$ drives the projection error to zero. Proposition~\ref{prop:inexact-pgd} then gives asymptotic feasibility on the training states and a bound on distance to the optimum governed by relative primal-step distortion. The following example, which uses refreshed targets, illustrates the analysis and the penalty bias~\eqref{eq:fused-bias-rate} on a problem whose constrained optimum is known.

\begin{example}[Learning a constrained optimal plan]
\label{ex:pointmass}
Consider the double integrator $x_{k+1}=Ax_k+Bu_k$, with $A=\left[\begin{smallmatrix}1&0.1\\0&1\end{smallmatrix}\right]$ and $B=\left[\begin{smallmatrix}0.005\\0.1\end{smallmatrix}\right]$, subject to $\|x_k\|_\infty\le15$ and $|u_k|\le1$. A network maps the state $x$ to a plan $U=(u_0,\dots,u_{N-1})$ with $N=20$, and the controller applies $u_0$. Performance is the discounted finite-horizon cost
\begin{equation}
J(U;x)=\sum_{k=0}^{N-1}\gamma^k\bigl(\|x_k-x_{\rm ref}\|_2^2+u_k^2\bigr)+\gamma^N\|x_N-x_{\rm ref}\|_{Q_N}^2,
\label{eq:pm-mpc}
\end{equation}
where $x_0=x$, $x_{\rm ref}=(14,0)$, $Q_N$ solves the discrete-time algebraic Riccati equation for the stage cost (the infinite-horizon LQR cost-to-go), and $\gamma=0.9$. The predicted states are linear in $U$, so $J$ is a strictly convex quadratic, $J(U;x)=\frac12\|U-U_{\rm unc}(x)\|_M^2+{\rm const}$, with $U_{\rm unc}(x)$ the unconstrained minimizer and $M\succ0$ the Hessian. This is the loss $\ell$ of Sec.~\ref{sec:problem-formulation} with labels $\bar u_i=U_{\rm unc}(x_i)$ and metric $M$. The feasible plans $\mathcal C(x)$ satisfy $\|x_k\|_\infty\le14$, a margin below the true bound, and $|u_k|\le1$ for all $k$, and end in the maximal control invariant set of that box. Each $\mathcal C(x)$ is a polytope, so projections in the metric $M$ are QPs, which we solve exactly, and the constrained optimum $v^\star$ stacks the optimal solutions $\argmin_{U\in\mathcal C(x_i)}J(U;x_i)$. The learner never sees $v^\star$: performance steps use only $\nabla J$, and safety steps use only projections of the network's own plans. We train on 2000 states with $K_t=\lceil c\log(t+2)\rceil$ refreshed safety steps per iteration, using Adam with a learning rate decaying from $10^{-3}$ to $10^{-5}$ over 4000 iterations, tripled for the safety steps.
\end{example}

\begin{table}[t]
\centering
\caption{Example~\ref{ex:pointmass}, mean over three seeds: gradient steps, RMS distance to $v^\star$, and closed-loop episodes leaving the true bounds. Fused with $\lambda=0$ minimizes the cost alone; interleaving is evaluated after 200 iterations.}
\label{tab:pointmass}
\footnotesize
\setlength{\tabcolsep}{3pt}
\begin{tabular}{@{}rccc c rccc@{}}
\toprule
\multicolumn{4}{c}{Fused} && \multicolumn{4}{c}{Interleaved}\\
\cmidrule{1-4}\cmidrule{6-9}
$\lambda$ & Steps & Dist. & Viol. && $c$ & Steps & Dist. & Viol.\\
\midrule
0    & 12k & 4.24 & 58\% && 1 & 1.2k & 0.18 & \textbf{0\%}\\
10   & 4k & 0.39 & 49\% && 2 & 2.0k & 0.15 & 1\%\\
30   & 4k & 0.16 & 24\% && 4 & 3.8k & 0.14 & \textbf{0\%}\\
100  & 4k & \textbf{0.11} & 10\% && & & & \\
300  & 4k & 0.13 & 2\% && & & & \\
1000 & 4k & 0.19 & 1\% && & & & \\
\bottomrule
\end{tabular}
\end{table}

Without access to $v^\star$, both methods move the learned plans from a distance of $4.24$ under the cost alone to within $0.11$--$0.39$ of $v^\star$ (Table~\ref{tab:pointmass}). Fused training trades safety against optimality through $\lambda$: small weights leave the plans outside the constraints, as the $O(1/\lambda)$ bias~\eqref{eq:fused-bias-rate} predicts, large weights slow convergence, and no weight is safe on every seed within 4000 steps. Interleaving with any safety budget is safe or nearly safe with fewer gradient steps, and a larger budget brings the plans closer to $v^\star$: the budget needs only to be large enough rather than tuned. Under plain gradient descent, which the analysis assumes, the contrast in distance to $v^\star$ is sharper: fused training stalls between $0.5$ and $1.1$, as its step must shrink with $\lambda$~\eqref{eq:fused-stability}, while interleaving keeps approaching $v^\star$ as the budget grows ($0.41$, $0.24$, and $0.16$ for $c=1,2,4$). Here corrections are dense, as the unconstrained plan is infeasible at about $70\%$ of the training states. If the policy instead outputs a single action and imitates saturated LQR, a correction is needed at only $1.5\%$ of the states, over $93\%$ of each safety block's output change falls on states that need none ($\kappa_S\ge0.96$), and more safety steps no longer reduce the distance to $v^\star$: Assumption~\ref{asm:safety-step-consistency} fails and interleaving stalls.

\section{Application to Autonomous Racing}
\label{sec:experiments}

Unlike Example~\ref{ex:pointmass}, racing has nonlinear dynamics, a sampling-based safety filter, and sparse corrections, so several assumptions of Sec.~\ref{sec:inexact-pgd} fail. We first summarize the main results and then examine the role of each component.

\subsection{Setup}
\label{sec:racing}

A simulated dynamic bicycle races on an L-shaped track under full-state feedback. A model predictive contouring control (MPCC)~\cite{liniger_optimization_2015} expert supplies demonstrations through DAgger~\cite{ross_dagger_2011} (expert weight $0.99^t$ at epoch $t$, buffer of up to 40\,000 samples). Policies with three width-128 ReLU hidden layers are trained for 300 epochs at learning rates that start at $10^{-3}$ and decay tenfold per 100 epochs. Both methods weight the safety loss by $\lambda$ and evaluate it on 16 states per minibatch, half from the batch and half from a fixed bank shared across methods. Fused training applies Adam to the sum of imitation and safety losses; interleaving applies one Adam imitation step followed by $K_t=\lfloor\log_2(t+2)\rfloor\le8$ SGD safety steps, with $t$ the DAgger epoch and targets refreshed before every step, so $\lambda$ scales its safety step size. The two methods thus also differ in the safety optimizer. The targets come from a predictive safety filter that returns the action closest to the proposed one among those certified safe: applying the action and then a backup policy that steers to the centerline while braking must keep a 15-step prediction within the track tightened by $0.02$ m. This is a projection onto a nonconvex certified set, which the filter approximates by enumerating 57 candidates around the proposed action.

Final policies are evaluated without a filter from 120 initial conditions per seed for up to 800 steps ($0.1$ s) or 12 laps. We report violating episodes ($|x_{\rm lat}|>0.55$ m or $|e_\psi|>\pi/2$) and the mean $\pm$ SD of seed-level lap times, counting only laps completed before the first violation. The expert laps in $5.171$ s but violates in $31/120$ episodes, and filtering it at runtime removes all violations at $5.219$ s. Plain BC laps in $5.144$ s but violates in $55/360$, and cloning the filtered expert still violates in $24/360$ at $5.189$ s, so safe labels alone do not give a safe policy. We seek the lap time of BC with the safety of the filtered expert.

\subsection{Main Results}
\label{sec:racing-main}

\begin{table}[t]
\centering
\caption{Ablations at $\lambda=10$ with 120 starts per seed: six seeds for the middle block and three for the others. The references use no safety loss. Lap time is mean $\pm$ SD of seed-level means (s). All methods refresh targets except Inter (log., fixed).}
\label{tab:mechanism-ablation}
\footnotesize
\setlength{\tabcolsep}{3pt}
\begin{tabular}{@{}l c r r@{}}
\toprule
Method & Safety updates & Violations & Lap time\\
\midrule
Plain BC on raw demos &-- &55/360 &$5.144\pm.002$\\
Plain BC on filtered demos &-- &24/360 &$5.189\pm.007$\\
\midrule
Fused &1 &149/720 &$5.209\pm.023$\\
Inter ($K=1$) &1 &84/720 &$5.142\pm.006$\\
Fused (compute matched) &$K_t$ &41/720 &$5.192\pm.014$\\
Inter (log., fixed) &$K_t$ &44/720 &$5.142\pm.006$\\
Inter (log., refreshed) &$K_t$ &\textbf{7/720} &$\mathbf{5.139\pm.006}$\\
\midrule
Fused, $\eta/(1+\lambda)$ &1 &271/360 &$5.336\pm.039$\\
Primal--dual, $\lambda_T\!\approx\!4.4$ &1 &56/360 &$5.201\pm.009$\\
Primal--dual, $\lambda_T\!\approx\!19$ &1 &46/360 &$5.178\pm.020$\\
Inter ($K=7$) &7 &9/360 &$5.150\pm.020$\\
\bottomrule
\end{tabular}
\end{table}

Refreshed interleaving with a logarithmic budget, Inter (log., refreshed), is the only configuration that approaches the safety of the filtered expert while keeping the lap time of unconstrained BC (Table~\ref{tab:mechanism-ablation}). At $\lambda=10$ it violates in $1.0\%$ of episodes at $5.139$ s, against $15.3\%$ at $5.144$ s for plain BC and $6.7\%$ at $5.189$ s for BC on filtered demonstrations. At this weight, every fused variant, whether with a fixed, compute-matched, or dual-adapted weight, is both slower ($5.18$--$5.21$ s) and less safe ($5.7$--$20.7\%$). Compared at its best weight, $\lambda=0.1$, fused training matches the lap time of BC ($5.136$ s) but still violates in $5.8\%$ of episodes, about six times as often as interleaving. The pattern holds across seeds (Table~\ref{tab:lambda-sweep}): at $\lambda\ge10$, no interleaved seed exceeds five violating episodes out of 120, whereas every fused seed has at least eight and one has 73. It also holds during training (Fig.~\ref{fig:training-curves}): fused training still averages $26$--$27\%$ violating episodes after epoch 200 at $\lambda\in\{10,100\}$ and $9\%$ when compute matched, whereas Inter (log., refreshed) records none after epoch 230. Fig.~\ref{fig:trajectories} shows the resulting closed-loop trajectories.

\begin{figure*}[t]
\centering
\includegraphics{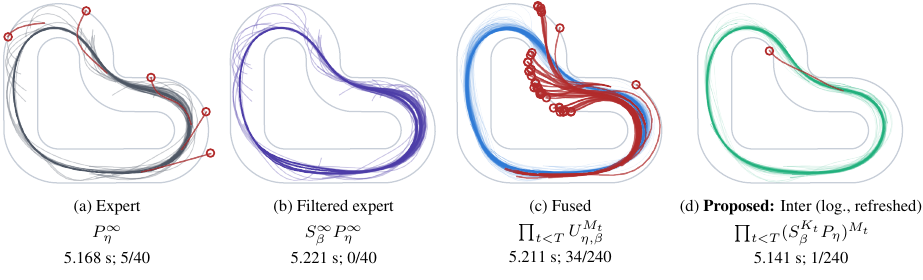}
\caption{Closed-loop trajectories on the 40-start illustration bank, with six learned-policy seeds at $\lambda=10$ and one run per expert configuration. Gray lines mark the track boundaries; each episode is drawn until its first violation, with the final $2.5$ s before it in red and the violation marked by a circle. Labels give the training update, mean lap time, and violating episodes, where $U_{\eta,\beta}(\theta)=\theta-\eta\nabla J(\theta)-\beta\nabla\Phi_{\mathcal C}(\theta)$ is the fused step, $T=300$ is the number of epochs, and $M_t$ is the number of minibatches in epoch $t$. This bank shares two starts with the primary 120-start bank.}
\label{fig:trajectories}
\end{figure*}

\begin{figure}[t]
\centering
\includegraphics[width=\columnwidth]{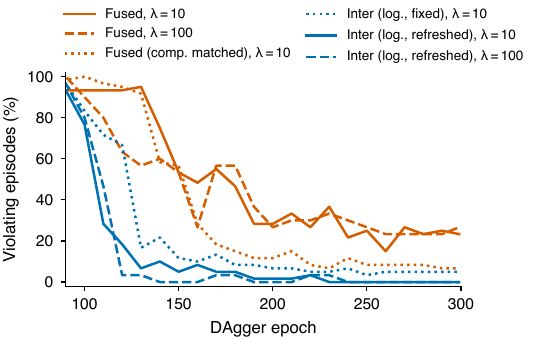}
\caption{Violating episodes during training (mean over seeds; six seeds at $\lambda=10$, three at $\lambda=100$), evaluated without a filter on ten starts every ten epochs. Before epoch 90, nearly every episode violates in all configurations.}
\label{fig:training-curves}
\end{figure}

\subsection{Ablations and Analysis}
\label{sec:mechanism-ablation}

\begin{table}[t]
\centering
\caption{Safety-weight sensitivity: violating episodes per seed (out of 120 each) and lap time, mean $\pm$ SD of seed-level means (s). The $\lambda=100$ runs and seeds four to six at $\lambda=10$ were added after the initial sweep.}
\label{tab:lambda-sweep}
\footnotesize
\setlength{\tabcolsep}{3pt}
\begin{tabular}{@{}c l r l r@{}}
\toprule
&\multicolumn{2}{c}{Fused Gradient}&\multicolumn{2}{c}{Interleaving (proposed)}\\
\cmidrule(lr){2-3}\cmidrule(l){4-5}
$\lambda$ & Violations & Lap time & Violations & Lap time\\
\midrule
0.01 & 10/7/18 & $5.149\pm.010$ & 8/9/12 & $5.139\pm.004$\\
0.1  & 10/3/8 & $5.136\pm.007$ & 20/3/27 & $5.146\pm.010$\\
1    & 33/2/7 & $5.159\pm.006$ & 8/20/30 & $5.151\pm.012$\\
10   & 8/25/12/73/18/13 & $5.209\pm.023$ & \textbf{0/0/0/0/5/2} & $\mathbf{5.139\pm.006}$\\
100  & 28/27/13 & $5.203\pm.011$ & \textbf{0/3/3} & $\mathbf{5.137\pm.002}$\\
\bottomrule
\end{tabular}
\end{table}

\begin{figure}[t]
\centering
\includegraphics[width=\columnwidth]{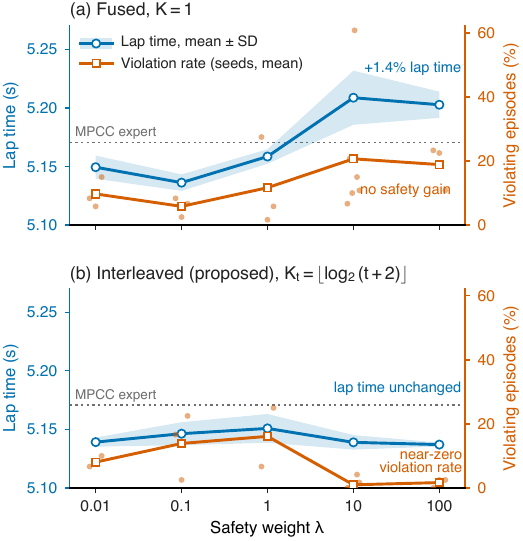}
\caption{Lap time and violation rate across safety weights. Bands show SD of seed-level lap means; dots show individual seed violation rates and squares their means. The shared dotted line marks the expert's $5.171$ s lap time and $25.8\%$ violation rate on the respective axes.}
\label{fig:lambda-sweep}
\end{figure}

\textit{Safety weight.} The weight plays different roles in the two methods (Table~\ref{tab:lambda-sweep}, Fig.~\ref{fig:lambda-sweep}). For fused training, $\lambda$ is a trade-off: raising it from $0.1$ to $10$ slows laps from $5.136$ to $5.209$ s, yet violations rise rather than fall. For interleaving, $\lambda$ scales the safety step and only sets how quickly violations vanish: lap times stay within $5.137$--$5.151$ s, while violations fall from $137/1080$ at $\lambda\le1$, where the decaying step leaves corrections incomplete within 300 epochs, to $13/1080$ at $\lambda\in\{10,100\}$. As in Example~\ref{ex:pointmass}, the interleaved budget needs only to be large enough.

\textit{Safety compute and adaptive weights.} Taking $K_t$ fused steps per minibatch, which matches the filter calls of interleaving, lowers fused violations to $41/720$, still well above $7/720$ (two-sided exact permutation test on seed-level counts, $p=0.0087$), and laps remain $0.05$ s slower. A primal--dual variant that raises $\lambda$ by dual ascent on the filter residual~\cite{tessler_reward_2019} ends at $\lambda\approx4.4$ or $19$ and behaves like a fixed weight, and shrinking the fused step to $\eta/(1+\lambda)$, as~\eqref{eq:fused-stability} requires, keeps the policy from learning to drive.

\textit{Refreshing and budget.} One safety step per minibatch gives $84/720$, a full logarithmic budget with fixed targets $44/720$, and with refreshed targets $7/720$. A constant block of $K=7$ steps does about as well as the logarithmic schedule ($9/360$), so within 300 epochs the total correction effort matters more than its growth.

\textit{Sparse corrections and data.} Racing lies outside the conditions of Theorem~\ref{thm:refreshed}. As in the single-action variant of Example~\ref{ex:pointmass}, corrections are sparse: the filter intervenes on only $6$--$7.5\%$ of the safety-block states throughout training, and $88$--$95\%$ of each safety block's squared output change falls on states with zero residual, so $\kappa_S\gtrsim0.94$ and $\kappa+\kappa_S<1$ fails. The gain over plain BC is nonetheless reproducible, 
and it acts through the data that DAgger collects. Fixing the final buffer of each run and retraining plain BC on it from scratch (three seeds), the buffer of an interleaved run ($K=7$) yields $13/360$ violating episodes at $5.139\pm.023$ s, close to the interleaved policies themselves ($9/360$), whereas the buffer of a plain-BC run yields $110/360$, twice the violations of the DAgger-trained BC policies, so retraining itself adds variance. The safety steps, whose step size is largest in early epochs, shape the policy while the buffer forms, and the states it then visits make the imitation targets safer. Fused and primal--dual training, with the same filter and expert, end less safe (Table~\ref{tab:mechanism-ablation}), although we did not retrain on their buffers.

\section{Conclusion}
\label{sec:conclusion}

Separating task optimization from safety correction turns imitation learning into inexact projected gradient descent in the space of policy outputs: when safety corrections are dense, the safety steps alone drive the policy toward the constrained optimum, and when they are sparse, as in racing, interleaving still yields policies that are safer than penalty-based training at the lap time of unconstrained imitation, through the data it collects, with a safety weight that needs only to be large enough. The guarantees concern the training states, so a runtime filter remains a sensible fallback. Extending them to sparse corrections, through architectures that localize safety updates or safety targets aligned with the expert's behavior, is left for future work.

\bibliographystyle{ieeetr}
\bibliography{references}

\end{document}